\documentclass{amsart}
\usepackage{amssymb,amsmath,amsfonts,latexsym,amsthm,mathrsfs}
\usepackage{amsaddr}
\usepackage{graphicx}
\numberwithin{equation}{section}

\newtheorem{theorem}{Theorem}[section]

\newtheorem{corollary}{Corollary}[section]
\renewcommand{\d}{\mathrm{d}}
\newcommand\vecwt{\tilde{\boldsymbol{w}}}

\newcommand{\veceh}{\hat{\boldsymbol{e}}}

\newcommand{\mt}{\tilde{m}}

\newcommand{\lambdat}{\tilde{\lambda}}

\newcommand{\figref}[1]{Figure~\ref{#1}}

\newcommand\beps{\mbox{\boldmath${\epsilon}$}}

\renewcommand{\d}{\mathrm{d}}

\newcommand{\Kc}{\mathcal{K}}

\newcommand{\Rc}{\mathcal{R}}

\newcommand{\bDelta}{\mbox{\boldmath${\Delta}$}}

\newcommand{\bcdot}{\mbox{\boldmath${\cdot}$}}

\newcommand{\btimes}{\mbox{\boldmath${\times}$}}

\newcommand{\bnabla}{\mbox{\boldmath${\nabla}$}}

\newcommand{\brho}{\mbox{\boldmath${\rho}$}}
\newcommand{\bsig}{\mbox{\boldmath${\sigma}$}}

\newcommand{\vecw}{\boldsymbol{w}}

\renewcommand{\d}{\mathrm{d}}

\newcommand{\bepsilon}{\mbox{\boldmath${\epsilon}$}}
\newcommand{\bsigma}{\mbox{\boldmath${\sigma}$}}
\renewcommand{\d}{\mathrm{d}}

\newcommand{\vecx}{\boldsymbol{x}}

\newcommand{\vece}{\boldsymbol{e}}

\newcommand{\vecE}{\boldsymbol{E}}

\newcommand{\vecD}{\boldsymbol{D}}
\newcommand{\vecJ}{\boldsymbol{J}}

\begin{document} 
\title[Spectral theory for discrete uniaxial polycrystalline materials]{Spectral theory of effective transport for discrete\\uniaxial polycrystalline materials}
	\author{N. Benjamin Murphy, Daniel Hallman, Elena Cherkaev, and Kenneth M. Golden}
	\address{Department of Mathematics, University of Utah,
		155 S 1400 E, Salt Lake City, UT 84112-0090}
	\maketitle

\begin{abstract}
We previously demonstrated that the bulk transport coefficients of uniaxial polycrystalline materials, including electrical and thermal conductivity, diffusivity, complex permittivity, and magnetic permeability,  
have Stieltjes integral representations involving spectral measures of self-adjoint random operators. The integral representations follow from resolvent representations of physical fields involving these self-adjoint operators, such as the electric field $\boldsymbol{E}$ and current density $\boldsymbol{J}$ associated with conductive media
with local conductivity $\boldsymbol{\sigma}$ and resistivity $\boldsymbol{\rho}$ matrices. In this article, we provide a discrete 
matrix analysis of this mathematical framework which parallels the continuum theory. We show that
discretizations of the operators yield real-symmetric random matrices which are composed of
projection matrices. We derive discrete resolvent representations for $\boldsymbol{E}$ and $\boldsymbol{J}$ 
involving the matrices which lead to eigenvector expansions of $\boldsymbol{E}$ and $\boldsymbol{J}$. We 
derive discrete Stieltjes integral representations for the components of the effective conductivity and resistivity 
matrices, $\boldsymbol{\sigma}^*$ and $\boldsymbol{\rho}^*$, involving spectral measures for the real-symmetric random matrices, which are given explicitly 
in terms of their real eigenvalues and orthonormal eigenvectors. We provide a projection 
method that uses properties of the projection matrices to show that the spectral measure can be computed by much 
smaller matrices, which leads to a more efficient and stable numerical algorithm for the computation of bulk transport 
coefficients and physical fields. We demonstrate this algorithm by numerically computing the spectral measure and current density for model 2D and 3D isotropic polycrystalline media with checkerboard microgeometry. 
\end{abstract}

\section{Introduction}\label{sec:Introduction}
In \cite{Murphy:2024:arXiv:PolycrystalTheoryContinuum} we formulated a mathematical framework 
to provide Stieltjes integral representations for the bulk transport coefficients for uniaxial polycrystalline materials 
\cite{Barabash:JPCM:10323,Gully:PRSA:471:2174,Milton:2002:TC}, including electrical and thermal conductivity, diffusivity, complex permittivity, and magnetic permeability, involving spectral measures of self-adjoint random operators. All of these transport phenomena are described locally by the same elliptic partial differential equation (PDE)~\cite{Milton:2002:TC}. For example, static electrical conduction \cite{Jackson-1999} is described by $\bnabla \bcdot \, (\bsigma \bnabla \phi)=0$ with electrical potential 
$\phi$, electric field $\vecE=-\bnabla\phi$, and local conductivity matrix $\bsigma$ given by \cite{Milton:2002:TC}
$\bsigma=R^T\text{diag}(\sigma_1,\sigma_2,\ldots,\sigma_2)R$ where $\sigma_1$ and $\sigma_2$ are real-valued, i.e., $\bsigma=R^T\text{diag}(\sigma_1,\sigma_2,\sigma_2)R$ for 3D and  $\bsigma=R^T\text{diag}(\sigma_1,\sigma_2)R$ for 2D. This is the general setting for 2D but introduces a uniaxial asymmetry for $d\ge3$,  i.e., the local conductivity along one of the crystal axes has the value $\sigma_1$, while the conductivity along all the other crystal axes
has the value $\sigma_2$. Here, $R$ is a rotation matrix that determines the arrangement and orientations of the crystallites comprising the polycrystalline material.

Consequently, the current and electric fields $\vecJ=\bsigma\vecE$ and $\vecE$ satisfy the electrostatic version of Maxwell's equations $\bnabla\bcdot\vecJ=\rho$ and $\bnabla\btimes\vecE=0$. 
In the long electromagnetic wavelength limit, the electric and displacement fields satisfy the quasistatic limit of Maxwell's equations \cite{Milton:2002:TC} $\bnabla\bcdot\vecD=\rho_b$ and $\bnabla\btimes\vecE=0$, where $\rho_b$ is the bound charge density \cite{Jackson-1999}, the displacement field is given by $\vecD=\bepsilon\vecE$, $\bepsilon=R^T\text{diag}(\epsilon_1,\epsilon_2,\ldots,\epsilon_2)R$, and the crystal permittivities $\epsilon_1(f)$ and $\epsilon_2(f)$ take \emph{complex values} which depend on the electromagnetic wave frequency $f$. In the remainder of this manuscript we formulate the problem of effective transport in terms of electrical conductivity, keeping in mind the broad applicability of the method.

Polycrystalline materials are solids that are composed of many 
crystallites of varying size, shape, and orientation. The crystallites 
are microscopic crystals which are held together by boundaries which can
be highly defective. 
Due to the highly irregular shapes of the crystallites and their
defective boundaries, on a microscopic level, the electromagnetic
transport properties of the medium can be quite erratic. As a
consequence, 
the derivatives in the transport equations may not hold in a classical sense and a weak formulation 
of the transport equation is necessary to provide a rigorous mathematical description of effective transport
for such materials \cite{Papanicolaou:RF-835,Golden:CMP-473}.

In \cite{Murphy:2024:arXiv:PolycrystalTheoryContinuum} we utilized the mathematical framework developed in \cite{Papanicolaou:RF-835,Golden:CMP-473} to show that the spectral measures underlying the Stieltjes integral 
representations for the bulk transport coefficients are associated with self-adjoint compositions of random and non-random projection operators. The random projection operators follow from writing  $\bsigma=R^T\text{diag}(\sigma_1,\sigma_2,\ldots,\sigma_2)R$ 
as $\bsigma=\sigma_1X_1+\sigma_2X_2$,  
where $X_1$ and $X_2$ are mutually orthogonal projection matrices that contain geometric information about the polycrystalline material via the rotation matrix $R$. A non-random projection 
operator arising in Stieltjes integrals for the components $\sigma^*_{jk}$, $i,j=1,\ldots,d$, of the effective electrical conductivity matrix $\bsigma^*$ is given by $\Gamma=\bnabla(\bnabla^*\bnabla)^{-1}\bnabla^*$, which is a projection onto the range of a generalized gradient operator $\bnabla$, where $\bnabla^*$ is the adjoint of $\bnabla$ and $d$ is the system dimension. The Stieltjes integrals involve spectral measures for the self-adjoint random operators $X_i\Gamma X_i$, $i=1,2$. The Stieltjes integrals for the components $\rho^*_{jk}$, 
$i,j=1,\ldots,d$, of the effective electrical resistivity matrix $\brho^*$ involve spectral measures for the 
self-adjoint random operators $X_i\Upsilon X_i$, $i=1,2$, where
$\Upsilon=\bnabla\btimes(\bnabla\btimes\bnabla\btimes)^{-1}\bnabla\btimes$ is a projection
onto the range of a generalized curl operator $\bnabla\btimes$.

Here develop a discrete mathematical framework that closely resembles the 
mathematical framework for the continuum setting 
\cite{Murphy:2024:arXiv:PolycrystalTheoryContinuum}, which
provides a rigorous way to numerically compute spectral measures,
effective parameters, and the physical fields $\vecE$ and $\vecJ$ 
for uniaxial polycrystalline materials. In particular, discretization of a continuous medium 
leads to a matrix analysis description of effective transport in composites which 
closely parallels the associated partial 
differential equation description for the continuum setting. Remarkably, the 
discrete formulation follows directly from just the fundamental theorem of 
linear algebra and known \cite{Huang:2019:AMSA.4.1} orthogonality properties of the domains, ranges, and 
kernels of discrete, finite difference representations of the curl, gradient,
and divergence operators.

In this discrete setting, the random operators $X_i\Gamma X_i$ and $X_i\Upsilon X_i$, $i=1,2$,
underlying the integral representations for the effective parameters
are represented by real-symmetric random matrices. The physical fields $\vecE$ and $\vecJ$ 
and the spectral measures are determined explicitly by their real eigenvalues and 
orthonormal eigenvectors. We provide a numerically efficient projection method to facilitate such
numerical computations, and use the method to compute the physical fields $\vecE$ and $\vecJ$,
spectral measures, and effective parameters for a model of isotropic polycrystalline media
with checkerboard microgeometry. In 2D, computed spectral
measures are in excellent agreement with known theoretical results, and the
computed values of the effective parameters fall within the rigorous bounds, as shown in 
\cite{Murphy:2024:arXiv:PolycrystalTheoryContinuum}.
%

\section{Stieltjes integrals for bulk transport coefficients for continuous media}
\label{sec:continuous_setting}
In this section, we briefly describe the \emph{analytic continuation method} 
for studying the effective transport properties of composite materials
\cite{Bergman:PRL-1285:44,Milton:APL-300,Golden:CMP-473,Golden:IMA-97}.
In particular, we review our application of the method to the setting of uniaxial 
polycrystalline materials given in \cite{Murphy:2024:arXiv:PolycrystalTheoryContinuum}.
This method has been used to obtain rigorous bounds on bulk
transport coefficients of polycrystalline media from partial knowledge
of the microstructure, such as the average orientation angle $\theta$ of each crystallite
\cite{Gully:PRSA:471:2174,Murphy:2024:arXiv:PolycrystalTheoryContinuum}. Examples of transport coefficients
to which this approach applies include
the complex permittivity, electrical and thermal conductivity,
diffusivity, and magnetic permeability.
To set ideas we focus on electrical conductivity, keeping the broad applicability of the method in mind. 
Consider a random medium with
local conductivity matrix $\bsig(x,\omega)$,
a spatially stationary random field
in $x\in\mathbb{R}^d$ and $\omega\in\Omega$, where $\Omega$ is the
set of realizations of the medium. 

The local conductivity matrix $\bsigma$ for uniaxial polycrystalline materials is given by \cite{Barabash:JPCM:10323,Gully:PRSA:471:2174,Milton:2002:TC,Murphy:2024:arXiv:PolycrystalTheoryContinuum}
$\bsigma=R^T\text{diag}(\sigma_1,\sigma_2,\ldots,\sigma_2)R$, i.e., $\bsigma=R^T\text{diag}(\sigma_1,\sigma_2,\sigma_2)R$ for 3D and  $\bsigma=R^T\text{diag}(\sigma_1,\sigma_2)R$ for 2D, where $R(\vecx,\omega)$ is a random rotation matrix determining the orientation of the crystallite with $\vecx$ in the interior for $\omega\in\Omega$ --- this is the general setting for 2D but introduces a uniaxial asymmetry for $d\ge3$. The orientation of each crystallite comprising the polycrystalline composite material is determined by a single angle $\theta(\vecx,\omega)$ and single rotation matrix for 2D, and $d$ angles $\theta_i(\vecx,\omega)$, $i=1,\ldots,d$ for dimension $d\ge3$, where $R$ is a composition of simple rotation 
matrices \cite{Murphy:2024:arXiv:PolycrystalTheoryContinuum}. Using the projection matrix $C= \text{diag}(\vece_1)$, where $\vece_k$ is the 
$k$th standard basis vector, the local conductivity tensor can be written as
\begin{align}\label{eq:sigma_two-phase}
	\bsig(x,\omega) = \sigma_1 \,X_1(x,\omega) 
	+ \sigma_2 \,X_2(x,\omega)\,.
\end{align}
where $X_1 = R^TCR$, $X_2 = R^T(I-C)R$ and $I$ is the identity matrix.
The random electric and current fields $\vecE(x,\omega)$ and $\vecJ(x,\omega)$ satisfy \cite{Jackson-1999}
\begin{align}\label{eq:Maxwells_Equations}
	\bnabla\times\vecE=0\,, 
	\quad
	\bnabla\cdot\vecJ = 0\,, 
	\quad
	\vecJ=\bsig\vecE\,.
\end{align}

We define \cite{Murphy:2024:arXiv:PolycrystalTheoryContinuum} the effective conductivity matrix $\bsig^*$ by 
\cite{Papanicolaou:RF-835}
$\langle \vecJ \rangle=  \bsig^* \langle \vecE \rangle$,
where $\langle\cdot\rangle$ is ensemble averaging over $\Omega$ 
or spatial average over all of $\mathbb{R}^d$ \cite{Golden:CMP-473}.
We define our coordinate system so that $\langle\vecE\rangle=\vecE_0=E_0\vece_k$.
A variational calculation 
\cite{Golden:CMP-473,Murphy:2024:arXiv:PolycrystalTheoryContinuum} 
establishes the energy constraint 
\begin{align}
\label{eq:E_variational}
\langle \vecJ\cdot\vecE_f\rangle=0
\end{align}
(Helmholtz's theorem) involving the fluctuating field 
$\vecE_f=\vecE-\vecE_0$. It then follows from 
$\langle \vecJ\cdot \vecE \rangle=\langle \vecJ\rangle\cdot \vecE_0$ that
$\bsigma^*$ is also given in terms of the energy functional 

\begin{align}
\label{eq:sigma*}
\langle \vecJ\cdot \vecE \rangle =  
\bsig^*\vecE_0\bcdot\vecE_0\,,
	\quad
	\vecE_0=E_0\vece_k\,.
\end{align}
Consequently, the method defines a homogeneous medium with
constant conductivity matrix $\bsig^*$ subject to a uniform applied electric 
field $\vecE_0=E_0\vece_k$ with constant electric field strength $E_0$ that behaves 
both macroscopically and energetically 
\emph{exactly the same} as the inhomogeneous polycrystalline material in the 
infinite volume limit \cite{Papanicolaou:RF-835}. 

The key step 
of the method 
is to obtain the following Stieltjes integral representation for $\bsig^*$ \cite{Bergman:PRC-377:9,Milton:APL-300,Golden:CMP-473,Milton:2002:TC},  
\begin{equation} \label{eq:F(s)}
	F_{jk}(s)=1 - \frac{\sigma^*_{jk}}{\sigma_2} =\int_0^1 {\frac{d\mu_{jk}(\lambda)}{s-\lambda}}\,,
	\quad s = \frac{1}{1-\sigma_1/\sigma_2}\,,
\end{equation}
where the $\mu_{jk}$, $j,k=1,\ldots,d$ are Stieltjes measure on $[0,1]$, the $\mu_{kk}$ are
\emph{positive} measures and the $\mu_{jk}$ for $j\ne k$ are \emph{signed} measures. 
In the variable $h=\sigma_1/\sigma_2$, 
$F(s)=\langle X_1\vecE\bcdot X_1\vece_k\rangle/(sE_0)$ is a \emph{Stieltjes function} 
\cite{Golden:PRL-3935,Murphy:JMP:063506}, playing the role of effective
electric susceptibility, $\boldsymbol{\chi}^*(s)=\bsig^*/\sigma_2-1=-F(s)$.
Equation \eqref{eq:F(s)} arises from a resolvent formula for the components of the electric field,
\begin{align}\label{eq:Electric_Field}
	X_1\vecE=s(sI-X_1 \Gamma X_1)^{-1}X_1\vecE_0,
	%
\end{align}
yielding 
$F_{jk}(s)=\langle[(sI-X_1 \Gamma X_1)^{-1}X_1\vece_j]\cdot X_1\vece_k\rangle$,
where $\Gamma  = -\bnabla (- \Delta )^{-1} \bnabla \cdot$ is a projection onto the range of the gradient operator $\bnabla$, where $\Delta=\bnabla\bcdot\bnabla=\nabla^2$ is the Laplacian operator. Formula \eqref{eq:F(s)}
is the spectral representation of the resolvent
and $\mu$ is the spectral measure of 
the self-adjoint operator
$X_1\Gamma X_1$
on $L^2(\Omega,P)$.
(Stieltjes integral representations for the $\sigma^*_{jk}$ can also be formulated in terms of the 
self-adjoint operator $X_2\Gamma X_2$ and the contrast variable $t=1/(1-\sigma_2/\sigma_1)=1-s$ \cite{Murphy:2024:arXiv:PolycrystalTheoryContinuum}.)

A key feature of equation \eqref{eq:F(s)} is that the component parameters, $\sigma_1$ and $\sigma_2$, 
in $s$ are separated from the geometrical information in the spectral measure $\mu_{jk}$. Information about the geometry enters
through the moments 
\begin{align}\label{eq:Moments_Fs}
	\mu_{jk}^n = 
	\int^1_0 \lambda^n d \mu_{jk}(\lambda)  =
	\langle [X_1\Gamma X_1]^n X_1 \vece_j \cdot X_1\vece_k \rangle.
\end{align}
The measure mass is given by $\mu_{jk}^0 = \langle X_1 \vece_j \cdot \vece_k \rangle$, which can be thought of as the ``mean orientation", or as a percentage of crystallites oriented in the $k^{\text{th}}$ direction. In general, higher order moments $\mu_{jk}^n$ depend on the $(n+1)$--point correlation function of the
medium \cite{Golden:CMP-473}.

The local resistivity matrix is the matrix inverse of the conductivity matrix, $\brho=\bsigma^{-1}$. Stieltjes integral representations for the components $\rho^*_{jk}$ of the effective resistivity matrix $\brho^*$ can be given in terms of the self-adjoint operators $X_1\Upsilon X_1$ and $X_2\Upsilon X_2$, where $\Upsilon=\bnabla\btimes(\bnabla\btimes\bnabla\btimes)^{-1}\bnabla\btimes$ is a 
projection onto the range of the curl operator $\bnabla\btimes$ \cite{Murphy:2024:arXiv:PolycrystalTheoryContinuum}. Here, the effective
resistivity matrix is defined by $\langle \vecE \rangle=  \brho^* \langle \vecJ \rangle$ and we define our coordinate 
system so that $\langle\vecJ\rangle=\vecJ_0=\vecJ_0\vece_k$. The energy constraint is given by
\begin{align}
	\label{eq:J_variational}
	\langle \vecJ_f\cdot\vecE\rangle=0
\end{align}
which leads to the energy functional representation for $\brho^*$
\begin{align}
\label{eq:rho*}
\langle \vecJ\cdot \vecE \rangle =  
\brho^*\vecJ_0\bcdot\vecJ_0\,,
\qquad
\vecJ=\vecJ_f+\vecJ_0,
\quad
\langle \vecJ \rangle=\vecJ_0=\vecJ_0\vece_k\,.
\end{align}

The resolvent formulas for the current density
are analogous to equation \eqref{eq:Electric_Field} \cite{Murphy:2024:arXiv:PolycrystalTheoryContinuum},
\begin{align}\label{eq:Current_Field}
X_1\vecJ=t(tI-X_1 \Upsilon X_1)^{-1}X_1\vecJ_0,
\quad
X_2\vecJ=s(sI-X_2 \Upsilon X_2)^{-1}X_2\vecJ_0.
\end{align}
Analogous to equation \eqref{eq:F(s)}, Stieltjes integral representations for the $\rho^*_{jk}$ are given 
by \cite{Murphy:2024:arXiv:PolycrystalTheoryContinuum}
\begin{equation} \label{eq:H(t)}
H_{jk}(t)=
1-\sigma_2\rho^*_{jk}=
\langle(tI-X_1\Upsilon X_1)^{-1}X_1\vece _j\bcdot\vece _k\rangle=
\int_0^1\frac{\d\kappa_{jk}(\lambda)}{t-\lambda}\,,
\quad 
t = \frac{1}{1-\sigma_1/\sigma_2}\,.
\end{equation}

These integral representations yields rigorous \emph{forward bounds} for the 
effective parameters of composites, given partial information on the 
microgeometry via the $\mu^n$ 
\cite{Bergman:PRL-1285:44,Milton:APL-300,Golden:CMP-473,Bergman:AP-78}.
One can also use the integral representations to obtain \emph{inverse
	bounds}, allowing one to use data about the electromagnetic response
of a sample, for example, to bound its structural parameters, such as
the mean orientation angle of the crystallites 
\cite{McPhedran:APA:19,McPhedran:MRSP:1990:195,Cherkaev:WRM-437,Cherkaev:IP-1203,
	Zhang:JCP-5390,Bonifasi-Lista:PMB-3063,Cherkaev:JBiomech-345,Day:JPCM-99,Golden:JBM:337}.

\section{Stieltjes integrals for bulk transport coefficients for discrete media}
\label{sec:discrete_setting}
In this section we adapt the discrete matrix analysis framework developed in \cite{Murphy:2015:CMS:13:4:825} that 
describes effective transport in discrete two-component composite materials, such as a random resistor network, to 
formulate a mathematical framework that provides discrete, matrix analysis versions of the results reviewed in 
Section \ref{sec:continuous_setting}. We show that a discretization of the operators
$X_1\Gamma X_1$ and $X_1\Upsilon X_1$, for example, lead to real-symmetric \emph{random matrices}. 
Here, $\Gamma$ and $\Upsilon$ are (non-random) projection matrices which depend only on the lattice 
topology and boundary conditions, and $X_1$ is a (random) projection matrix 
which determines the crystallite arrangement and orientation angles of the polycrystalline material.

In Section \ref{sec:field_decomp}, we prove the existence of transport fields $\vecE$ and $\vecJ$ satisfying 
discrete forms of the following systems of equations
\begin{align}   \label{eq:Maxwells_Equations_E}  
&\bnabla \times\vecE =0, \quad
\bnabla \bcdot\vecJ=0,\quad
\vecJ=\bsigma\vecE ,\quad
\vecE =\vecE _0+\vecE _f, \quad
\langle\vecE \,\rangle=\vecE _0\,, 
 \\
 \label{eq:Maxwells_Equations_J}
    &\bnabla \times\vecE =0, \quad
    \bnabla \bcdot\vecJ=0, \quad
    \vecE =\brho\vecJ,\quad
    \vecJ=\vecJ_0+\vecJ_f,\quad
    \langle\vecJ\,\rangle=\vecJ_0\,.
\end{align}
In Section \ref{sec:resolvent_reps}, we solve these systems of equations by deriving resolvent representations for $X_i\vecE$ and $X_i\vecJ$, $i=1,2$, in \eqref{eq:Electric_Field} and \eqref{eq:Current_Field}, with 
$\vecE=X_1\vecE+X_2\vecE$ and $\vecJ=X_1\vecJ+X_2\vecJ$, in terms of the real-symmetric matrices
$X_i\Gamma X_i$ and $X_i\Upsilon X_i$, $i=1,2$. We show that these resolvent representations have summation formulas given explicitly in terms of the real eigenvalues and orthonormal eigenvectors of the real-symmetric matrices $X_1\Gamma X_1$ and $X_1\Upsilon X_1$
\begin{align}\label{eq:discrete_resolvent}
X_1 \vecE = sE_0\sum_j \frac{(\vecw_i \cdot X_1 \hat \vece_k)}{s-\lambda_i} \vecw_i\,, 
\quad
X_1 \vecJ =t\vecJ_0\sum_i \frac{(\vecwt_i \cdot X_1 \hat \vece_k)}{t-\lambdat_i} \vecwt_j\,,
\end{align}
where the $(\lambda_i,\vecw_i)$ and $(\lambdat_i,\vecwt_i)$ are the eigenvalues and eigenvectors of the
matrices $X_1\Gamma X_1$ and $X_1\Upsilon X_1$, respectively.

In Section \ref{sec:Stieltjes_integrals}, we show that the integrals in equations \eqref{eq:F(s)} 
and \eqref{eq:H(t)}, for example, also have summation formulas \cite{Murphy:2015:CMS:13:4:825},
\begin{align}
\label{eq:Discrete_E_F(s)}
&F_{kk}(s)=\sum_i\left\langle
\frac{m_i}{s-\lambda_i}\right\rangle\,,
\quad
m_i = |\vecw_i\,\cdot\,X_1 \hat{\vece}_k|^2,
\qquad
&H_{kk}(t)=\sum_i\left\langle
\frac{\mt_i}{t-\lambdat_i}\right\rangle\,,
\quad
\mt_i = |\vecwt_i\,\cdot\,X_1 \hat{\vece}_k|^2,
\end{align}
where $\hat{\vece}_k$ plays the role of a standard basis vector on the
lattice. Then the discrete spectral measure $\mu_{kk}$, for example, is given 
in terms of the eigenvalues $\lambda_i$ and 
orthonormal eigenvectors $\vecw_i$ of the matrix $X_1\Gamma X_1$
\cite{Murphy:2015:CMS:13:4:825},   
\begin{align}\label{eq:Stieltjes_F_Discrete}
	\mu(d\lambda)=\langle Q(d\lambda)\hat{\vece}_k\cdot\hat{\vece}_k\rangle,
	\quad 
	Q(d\lambda)=\sum_i\delta_{\lambda_i}(d\lambda)X_1Q_i\,,
	\quad
	Q_i=\vecw_i\vecw_i^T\,.
\end{align}
Here,  $Q(d\lambda)$ is the 
projection valued measure associated
with the matrix $X_1\Gamma X_1$, $\delta_{\lambda_i}(d\lambda)$ is the delta measure centered at
$\lambda_i$, and the matrix $Q_i=\vecw_i\vecw_i^T$ is a projection onto the
eigenspace spanned by $\vecw_i$ \cite{Murphy:2015:CMS:13:4:825}.

\subsection{Existence of transport fields for discrete media}
\label{sec:field_decomp}
In this section, we prove that there exists transport fields $\vecE$ and $\vecJ$ that satisfy discrete versions 
of equations \eqref{eq:Maxwells_Equations_E} and \eqref{eq:Maxwells_Equations_J}. Specifically, 
in Theorem \ref{thm:existence_E}, we prove the existence of an electric field $\vecE$ 
satisfying a discrete version of the system of equations in \eqref{eq:Maxwells_Equations_E}. 
In Corollary \ref{cor:existence_J}, we prove that there exists a current field $\vecJ$ 
satisfying a discrete version of the system of equations in \eqref{eq:Maxwells_Equations_J}.

Towards this goal,
we follow \cite{Huang:2019:AMSA.4.1} and consider finite difference 
representations of the partial differential operators $\partial_i\to C_i$,
$i=1\ldots,d$, where $d$ denotes dimension. The matrices $C_i$ depend on boundary
conditions which, without loss of generality, we take to be periodic boundary 
conditions. Denote the matrix representation 
of the gradient operator $\bnabla$ (using Matlab vertical block notation) 
by $\nabla=[C_1;\ldots;C_d]$. The discretization of the divergence operator is 
given by $-\nabla^T$ and the discrete curl operator $\bnabla\btimes$ is given by
\cite{Huang:2019:AMSA.4.1}
\begin{align}\label{eq:block_C}
C&=\left[
\begin{array}{rrr}
	O&-C_3&C_2\\
	C_3&O&-C_1\\  
	-C_2&C_1&O
\end{array}
\right]
\;\text{ for 3D,} 
\\
C&=[-C_2,C_1]
\;\text{ for 2D.} 
\notag
\end{align}
The operators $C_i$, $i=1,2,3$, in \eqref{eq:block_C} are normal 
and commute with each other \cite{Huang:2019:AMSA.4.1},
\begin{align}
	C_i^T C_j=C_jC_i^T \text{ and } C_iC_j=C_jC_i, \text{ for } i,j=1,2,3.
\end{align}

We now summarize some useful identities relating the discrete representations
of the gradient, divergence, and curl operators which follow from these properties 
of the matrices $C_i$ \cite{Huang:2019:AMSA.4.1},
\begin{align}\label{identities}
	&\Delta=\bnabla\bcdot\bnabla\to-\nabla^T\nabla\,,
	\\
	&\bDelta = \text{diag}(\Delta,\ldots,\Delta)
	\to I_d\otimes(\nabla^T\nabla)\,,
	\notag\\
	&\bnabla\times\bnabla\times\to C^TC\,,
	\notag\\
	&\bnabla\times\bnabla\times=\bnabla(\bnabla\bcdot)-\bDelta
	\to -\nabla\nabla^T + I_d\otimes(\nabla^T\nabla)\,,
	\notag\\
	&\bnabla\bcdot(\bnabla\times)\to -\nabla^TC^T=0\,,
	\notag\\
	&\bnabla\times\bnabla\to C\nabla=0\,,
	\notag
\end{align}
where $\otimes$ denotes the Kronecker product.
The last two identities $\nabla^TC^T=0$ 
and $C\nabla=0$ in equation \eqref{identities} indicate that 
%
\begin{align}\label{subspaces}
	\Rc(C^T)\subseteq\Kc(\nabla^T)\,,
	\quad
	\Rc(\nabla)\subseteq\Kc(C)
\end{align}

Consequently, the discrete form of the system of equations in \eqref{eq:Maxwells_Equations_E} 
and the energy constraint in \eqref{eq:E_variational} is
\begin{align}\label{eq:Maxwells_Equations_discrete_decomp}
	C\vecE=0, 
	\quad -\nabla^T \vecJ=0, 
	\quad \vecJ=\beps \vecE,
	\quad
	\vecE=\vecE_f+\vecE_0\,,
	\quad
	C\vecE_f=0\,,    
	\quad
	\langle\vecJ\bcdot\vecE_f\rangle=0\,,
	\quad
	\langle\vecE_f\rangle=0\,,
\end{align}
where in this finite size discrete setting, 
$\langle\cdot\rangle$ denotes volume average followed by ensemble average
\cite{Murphy:2015:CMS:13:4:825,Murphy:JMP:063506}.
To set notation, denote by $\Rc(B)$ and $\Kc(B)$ 
the range and kernel (null space) of a matrix $B$, respectively 
\cite{Horn_Johnson-1990}.  
Therefore, we seek to prove that there exists a vector $\vecE$ 
satisfying $\vecE\in\Kc(C)$ such that $\vecE=\vecE_f+\vecE_0$, where 
$\vecE_f\in\Kc(C)$ and $\langle \vecE_f\rangle=0$ with $\vecE_0\in\Kc(C)$, so 
$\langle\vecE\rangle=\vecE_0$. 
Moreover, we seek to find a vector 
$\vecJ\in\Kc(\nabla^T)$ satisfying $\vecJ=\beps\vecE$ and 
$\langle\vecJ\bcdot\vecE_f\rangle=0$ and $\vecJ_0\in\Kc(\nabla^T)$. 
Given the formulas in equation \eqref{subspaces}, we focus on finding 
``potentials'' $\varphi$ and $\psi$ such that $\vecE_f=\nabla\varphi$  
and $\vecJ_f=C^T\psi$.

The last two identities \eqref{identities} provide a relationship between 
rank and kernel of the operators $C$, $\nabla$, and their transposes. 
The fundamental theorem of linear algebra \cite{Horn_Johnson-1990} provides a relationship 
between the range of a matrix $A$ and the kernel of it's transpose $A^T$, 
which will be useful later in this section.
\begin{theorem}[Fundamental theorem of linear algebra]
	\label{FTLA}
	Let $A$ be a real valued matrix of size $m\times n$ then
	\begin{align} \label{eq:FTLA}
		\mathbb{R}^m=\Rc(A)\oplus\Kc(A^T)\,,
		\qquad
		\mathbb{R}^n=\Rc(A^T)\oplus\Kc(A)\,,
	\end{align}
	where $\oplus$ indicates $\Rc(A)$ is orthogonal 
	to $\Kc{(A^T)}$, i.e., $\Rc(A)\perp\Kc{(A^T)}$, 
	for example.
\end{theorem}

Applying Theorem \ref{FTLA} to the matrices $\nabla$ 
and $C^T$ indicates that
$\mathbb{R}^m=\Rc(\nabla)\oplus\Kc{(\nabla^T)}$ and 
$\mathbb{R}^m=\Rc(C^T)\oplus\Kc{(C)}$. Therefore,
from equation \eqref{subspaces} we have that
divergence-free fields are orthogonal to gradients (curl-free fields)
and curl-free fields are orthogonal to $\Rc(C^T)$ (divergence-free fields). 
This is a discrete version of the Helmholz Theorem, which states that curl-free 
and divergence-free fields 
(or, in other words, the gradient and cycle spaces) are mutually 
orthogonal. This also establishes the important relationship 
%
\begin{align}\label{eq:perp_C_grad}
	\Rc(C^T)\perp\Rc(\nabla)\,.    
\end{align}

Orthogonal bases can be given for each of 
the mutually orthogonal spaces in equation \eqref{eq:FTLA} through the singular
value decomposition (SVD) \cite{Horn_Johnson-1990} of a matrix $A=U\Sigma V^T$,
which also provides important information relating the matrices $C$, $\nabla$, etc. 
Here $U$ and $V$ are 
orthogonal matrices of size $m\times m$ and $n\times n$ satisfying $U^T U=UU^T=I_m$ 
and $V^T V=VV^T=I_n$, where $I_m$ is the identity matrix of size $m$. Moreover, $\Sigma$
is a diagonal matrix of size $m\times n$ with diagonal components consisting of the 
positive \emph{singular values} $\nu_i$, $i=1,\ldots,n$, of the matrix $A$ satisfying 
$\nu_1\ge\nu_2\ge\cdots\ge\nu_\rho>0$ and $\nu_{\rho+1}=\nu_{\rho+2}=\cdots=\nu_n=0$,
where $\rho$ is the rank of $A$. Writing the matrix $\Sigma$ in block form we have
\begin{align}\label{eq:block_Sigma}
	\Sigma=\left[
	\begin{array}{cc}
		\Sigma_1&O_1\\
		O_1^T&O_2\\  
		O_3&O_4
	\end{array}
	\right]\,,
\end{align}
where $\Sigma_1$ is a diagonal matrix of size $\rho\times\rho$ with diagonal
consisting of the strictly positive singular values, $O_1$ and $O_2$, are matrices 
of zeros of size $\rho\times(n-\rho)$ and $(n-\rho)\times(n-\rho)$, and the bottom block 
of zeros $[O_3,O_4]$ is of size $(m-n)\times n$.

Write the matrices $U$ and $V$ in block form as
$U=[U_1,U_0,U_2]$ and $V=[V_1,V_0]$, where $U_1$ and $V_1$ are the 
columns of $U$ and $V$ corresponding to the strictly positive singular values in $\Sigma_1$,
$U_0$ and $V_0$ correspond to the the singular values with value zero, $\nu_i=0$, and $U_2$
corresponds to the bottom block of zeros $[O_3,O_4]$ in $\Sigma$. Taking in account all the
blocks of zeros in $\Sigma$, we can write $A=U_1\Sigma_1V_1^T$. 
We now state a well known fact about the SVD of the matrix $A$ \cite{Horn_Johnson-1990}.
\begin{align}\label{svd_col_space}
	\Rc(A)=\text{Col}(U_1),
	\quad
	\Kc(A)=\text{Col}(V_0),
	\quad
	\Rc(A^T)=\text{Col}(V_1),
	\quad
	\Kc(A^T)=\text{Col}([U_0,U_2]),
\end{align}
where $\text{Col}(B)$ denotes the column space of the matrix $B$, i.e., 
the space spanned by the columns of $B$.

Applying the SVD
to the matrices $\nabla=U^\times\Sigma^\times[V^\times]^T$ 
and $C^T=U^\bullet\Sigma^\bullet[V^\bullet]^T$ and using the orthogonality 
of the columns of the matrices $U_1$ and $V_1$ and 
the invertibility of $\Sigma_1$, from $C\nabla=0$ in 
\eqref{identities} we have $[U^\bullet]^T U_\times=0$, and similarly 
$\nabla^TC^T=0$ implies $[U^\times]^T U_\bullet=0$. The formula
$[U^\bullet]^T U_\times=0$ is a restatement of equation 
\eqref{eq:perp_C_grad}. Writing 
$U^\times=[U_1^\times,U_0^\times,U_2^\times]$ and 
$U^\bullet=[U_1^\bullet,U_0^\bullet,U_2^\bullet]$ we have 
established that $\Rc(\nabla)=\text{Col}(U_1^\times)$,
$\Rc(C^T)\subseteq\Kc(\nabla^T)
=\text{Col}([U_0^\times,U_2^\times])$. Also, since
$\Rc(C^T)=\text{Col}(U_1^\bullet)$ and 
$\Rc(C^T)\perp\Rc(\nabla)$, we can 
write
\begin{align}
	U^\times=[U_1^\times,U_0^{\times\bullet},U_1^\bullet]\,,
	\quad
	U^\bullet=[U_1^\bullet,U_0^{\times\bullet},U_1^\times]\,,
\end{align}
where the columns of $U_0^{\times\bullet}$ are orthogonal to both
the $\Rc(C^T)$ and the $\Rc(\nabla)$. Since
$U^\times [U^\times]^T=I_m$ we have the \emph{resolution of 
	the identity}
\begin{align}\label{eq:res_of_identity}
	\Gamma_\times+\Gamma_0+\Gamma_\bullet
	=I_m,
	\qquad
	\Gamma_\times=U_1^\times[U_1^\times]^T\,, 
	\quad
	\Gamma_\bullet=U_1^\bullet[U_1^\bullet]^T\,,
	\quad
	\Gamma_0=U_0^{\times\bullet}[U_0^{\times\bullet}]^T\,,
\end{align}
where $\Gamma_\times$, $\Gamma_\bullet$, and $\Gamma_0$, 
are mutually orthogonal projection matrices onto $\Rc(\nabla)$, 
$\Rc(C^T)$ and the orthogonal complement of 
$\Rc(\nabla)\cup\Rc(C^T)$. We are now ready to 
state the main result of this section as the following theorem.
\begin{theorem}\label{thm:existence_E}
Let the electric and current fields $\vecE$ and $\vecJ$ satisfy
\begin{align}\label{eq:Maxwells_Equations_discrete_E}
	C\vecE=0, 
	\quad -\nabla^T \vecJ=0, 
	\quad \vecJ=\beps \vecE.
\end{align}
Then, there exists a ``potential'' $\varphi$ and vector $\vecE_0$ such that
$\vecE=\vecE_f+\vecE_0$, where $\vecE_f=\nabla\varphi$,
$C\vecE_f=0$, 
$\langle\vecJ\bcdot\vecE_f\rangle=0$, $\langle\vecE_f\rangle=0$, and
$I_d\otimes(\nabla^T\nabla))\vecE_0=0$, thus $\vecE_0$ is an arbitrary constant.
\end{theorem}
\begin{proof}
Let $\vecE$ be the solution to equation \eqref{eq:Maxwells_Equations_discrete_E}. 
From the resolution of the identity in equation \eqref{eq:res_of_identity}
we have $\vecE=(\Gamma_\times+\Gamma_0+\Gamma_\bullet)\vecE$.
Since $\Gamma_\bullet$ projects onto the $\Rc(C^T)$,
$\mathbb{R}^m=\Rc(C^T)\oplus\Kc{(C)}$, and 
$\vecE\in\Kc{(C)}$ we have $\Gamma_\bullet \vecE=0$. Defining
$\vecE_f=\Gamma_\times \vecE$, since 
$U_1^\times=\nabla V_1^\times[\Sigma_1^\times]^{-1}$, we can write 
$\vecE_f=\nabla\varphi$, where
$\varphi=V_1^\times[\Sigma_1^\times]^{-1}[U_1^\times]^T\vecE$.

Define 
$\vecE_0=\Gamma_0\vecE$. Since $\Gamma_\times \vecE_0=0$, $\Gamma_\times$ is a 
projection onto $\Rc(\nabla)$, and 
$\mathbb{R}^m=\Rc(\nabla)\oplus\Kc{(\nabla^T)}$, we have 
$\vecE_0\in\Kc(\nabla^T)$. Similarly, since $\Gamma_\bullet \vecE_0=0$,
$\Gamma_\bullet$ is a 
projection onto $\Rc(C^T)$, and $\mathbb{R}^m=\Rc(C^T)\oplus\Kc{(C)}$, we have 
$\vecE_0\in\Kc(C)$. Since $\vecE_0\in\Kc(C)\cap\Kc(\nabla^T)$
we have from \eqref{identities} that
\begin{align}
\label{eq:discrete_solution_E0}	
0=C^TC\vecE_0 
= (-\nabla\nabla^T + I_d\otimes(\nabla^T\nabla))\vecE_0
= I_d\otimes(\nabla^T\nabla))\vecE_0,
\end{align}
which determines $\vecE_0$ and establishes that the solution to \eqref{eq:Maxwells_Equations_discrete_E}
can be written as $\vecE=\nabla\varphi+\vecE_0$.
Equation \eqref{eq:discrete_solution_E0} implies that each dimensional component 
$(\vecE_0)_i$, $i=1,\ldots,d$, which is a vector in this discrete setting \cite{Murphy:2015:CMS:13:4:825},
satisfies $\nabla^T\nabla(\vecE_0)_i$. Therefore,
\begin{align}
0=\nabla^T\nabla(\vecE_0)_i\bcdot(\vecE_0)_i=
\nabla(\vecE_0)_i\bcdot\nabla(\vecE_0)_i=
\|\nabla(\vecE_0)_i\|^2\,,
\end{align}
where $\|\cdot\|$ denotes $\ell^2$-norm. It follows that $(\vecE_0)_i$ is constant for each $i=1,\ldots,d$,
hence $\vecE_0$ is an arbitrary constant vector.

From $\Rc(\nabla)\subseteq\Kc(C)$ in equation
\eqref{subspaces} and $\vecE_f=\nabla\varphi$ we have 
$C\vecE_f=C\nabla\varphi=0$. We also have from $\nabla^T\vecJ=0$ that
$\vecJ\bcdot \vecE_f=\vecJ\bcdot\nabla\varphi=\nabla^T\vecJ\bcdot\varphi=0$.
Finally, the volume average of $\nabla\varphi$ is a telescoping sum, 
so $\langle \vecE_f\rangle=0$. This concludes our proof of the theorem.
\end{proof}

\begin{corollary}
\label{cor:existence_J}
Let the electric and current fields $\vecE$ and $\vecJ$ satisfy
\begin{align}\label{eq:Maxwells_Equations_discrete_J}
	C\vecE=0, 
	\quad -\nabla^T \vecJ=0, 
	\quad \vecE=\brho \vecJ.
\end{align}
Then, there exists a ``potential'' $\psi$ and vector $\vecJ_0$ such that
$\vecJ=\vecJ_f+\vecJ_0$, where $\vecJ_f=C^T\psi$, $\nabla^T \vecJ_f=0$, 
$\langle\vecJ_f\bcdot\vecE\rangle=0$,  $\langle\vecJ_f\rangle=0$, and
$I_d\otimes(\nabla^T\nabla))\vecJ_0=0$, thus $\vecJ_0$ is an arbitrary constant.
\end{corollary}
\begin{proof}
Let $\vecJ$ be the solution to equation \eqref{eq:Maxwells_Equations_discrete_J}. 
From the resolution of the identity in equation \eqref{eq:res_of_identity}
we have $\vecJ=(\Gamma_\times+\Gamma_0+\Gamma_\bullet)\vecJ$.
Since $\Gamma_\times$ projects onto the $\Rc(\nabla)$,
$\mathbb{R}^m=\Rc(\nabla)\oplus\Kc(\nabla^T)$, and 
$\vecJ\in\Kc(\nabla^T)$ we have $\Gamma_\times \vecJ=0$. Defining 
$\vecJ_f=\Gamma_\bullet \vecJ$, since 
$U_1^\bullet=C^T V_1^\bullet[\Sigma_1^\bullet]^{-1}$, we can write 
$\vecJ_f=C^T\psi$, where
$\psi=V_1^\bullet[\Sigma_1^\bullet]^{-1}[U_1^\bullet]^T\vecJ$. Define 
$\vecJ_0=\Gamma_0 \vecJ$. Exactly the same as in the proof of Theorem
\ref{thm:existence_E}, we have $I_d\otimes(\nabla^T\nabla))\vecJ_0=0$,
which determines $\vecJ_0$ to be an arbitrary constant vector and establishes that the solution to \eqref{eq:Maxwells_Equations_discrete_J}
can be written as $\vecJ=C^T\psi+\vecJ_0$.
From $\Rc(C^T)\subseteq\Kc(\nabla^T)$ in equation
\eqref{subspaces} and $\vecJ_f=C^T\psi$ we have 
$\nabla^T\vecJ_f=\nabla^TC^T\psi=0$. Exactly the same as in the proof 
of Theorem \ref{thm:existence_E}, we also have $\vecJ_f\bcdot\vecE=0$ and 
$\langle \vecJ_f\rangle=0$. This concludes our proof of the theorem.
\end{proof}
We end this section by noting that in the full rank setting, where 
$\Sigma$ has all strictly positive singular values, so $\Sigma_1=\Sigma$, 
then 
\begin{align}\label{eq:GammaCurlDiscrete}
	\Gamma_\times = \nabla(\nabla^T\nabla)^{-1}\nabla^T\,,
	\qquad
	\Gamma_\bullet = C^T(CC^T)^{-1}C\,.
\end{align}
The original formulations of this mathematical framework was given in terms of these projection matrices, or continuum generalizations \cite{Golden:CMP-473,Murphy:2015:CMS:13:4:825}. 
The formulation given in this section
generalizes the discrete setting to cases where the matrix gradient is rank deficient,
such as the case of periodic boundary conditions. To simplify notation to that used in \cite{Murphy:2024:arXiv:PolycrystalTheoryContinuum}, 
for the remainder of the manuscript we will denote $\Gamma=\Gamma_\times$ and $\Upsilon=\Gamma_\bullet$.

\subsection{Resolvent representations for transport fields}
\label{sec:resolvent_reps}
In this section, we utilize the results of Theorem \ref{thm:existence_E} to explicitly solve the system of 
equations in  \eqref{eq:Maxwells_Equations_discrete_E}  for $\vecE$. We accomplish this by deriving 
resolvent formulas for the fields $X_i\vecE$, $i=1,2$, in \eqref{eq:Electric_Field} in terms of the real-symmetric matrices $X_i\Gamma X_i$, $i=1,2$. The electric field can then be constructed using 
$\vecE=X_1\vecE+X_2\vecE$ and the current density can be constructed using 
$\vecJ=\sigma_1X_1\vecE+\sigma_2X_2\vecE$. Similarly, we utilize the results of Corollary 
\ref{cor:existence_J} to explicitly solve the system of 
equations in  \eqref{eq:Maxwells_Equations_discrete_J} for $\vecJ$ by deriving 
resolvent formulas for the fields $X_i\vecJ$ , $i=1,2$, in \eqref{eq:Current_Field} 
in terms of the real-symmetric matrices $X_i\Upsilon X_i$, $i=1,2$. The current field can then be constructed 
using $\vecJ=X_1\vecJ+X_2\vecJ$ and the electric field can be constructed using
$\vecE=\sigma_1^{-1}X_1\vecJ+\sigma_2^{-1}X_2\vecJ$. These resolvent formulas provide eigenvector expansions
for the physical fields $\vecE$ and $\vecJ$, which can be used to numerically calculate them, 
which we do in Section \ref{sec:Numerical_Results}.

We provide a detailed derivation of the resolvent formula shown in  \eqref{eq:Electric_Field} for the setting 
where the gradient matrix $\nabla$ is full rank and rank deficient. Other resolvent formulas are obtained in
an analogous manner. From Theorem \ref{thm:existence_E} and equation \eqref{eq:Maxwells_Equations_discrete_E}
we can apply the operator $\nabla(\nabla^T\nabla)^{-1}$ to the formula $-\nabla^T\vecJ=0$ to obtain
$\Gamma\vecJ=0$, where $\Gamma=\nabla(\nabla^T\nabla)^{-1}\nabla^T$ is a projection onto $\Rc(\nabla)$.
In the full rank setting, the SVD for $\nabla=U_1\Sigma_1 V_1^T$, where all singular values have strictly positive value.
In this case, the square $n\times n$ diagonal matrix $\Sigma_1$ is invertible, $U_1$ is an $m\times n$ with left inverse satisfying 
$U_1^T U_1=I_m$, and $V_1$ is an $n\times n$ unitary matrix satisfying $V_1V_1^T=V_1^T V_1=I_n$ \cite{Horn_Johnson-1990}.
Using these matrix properties, we have $\Gamma=\nabla(\nabla^T\nabla)^{-1}\nabla^T=U_1U_1^T$
\cite{Murphy:ADSRF-2019}.

When $\nabla$ is rank deficient, one or more of the singular values have value 0. We can still write $\nabla=U_1\Sigma_1 V_1^T$  \cite{Horn_Johnson-1990,Murphy:ADSRF-2019}. However, now the diagonal matrix $\Sigma_1$ is smaller than in the full rank setting but it is still invertible, 
and both the matrices $U_1$ and $V_1$ just have left inverses satisfying $V_1^T V_1=I_n$ and $U_1^T U_1=I_m$ \cite{Horn_Johnson-1990,Murphy:ADSRF-2019}. The key observation is that we can use these matrix properties to again write $-\nabla^T\vecJ=0$ as $\Gamma\vecJ=0$ where we \emph{define} $\Gamma=U_1U_1^T$, and $\Gamma$ is still a projection on to $\Rc(\nabla)$ \cite{Horn_Johnson-1990,Murphy:ADSRF-2019}.

Writing equation \eqref{eq:sigma_two-phase} as $\bsigma=\sigma_2(1-X_1/s)$,
using the results of Theorem \ref{thm:existence_E} to write $\Gamma\vecE_f=\vecE_f$, $\Gamma\vecE_0=0$, and
$\vecJ=\sigma_2(1-X_1/s)(\vecE_f+\vecE_0)$, the formula $\Gamma\vecJ=0$ can be written as 
$\vecE_f-\Gamma X_1\vecE/s=0$ \cite{Murphy:2015:CMS:13:4:825}. Adding $\vecE_0$ to both sides then applying 
$X_1$ to the left of both sides, using $X_1^2=X_1$, and rearranging yields 
\begin{align}
(sI-X_1\Gamma X_1)X_1\vecE=sX_1\vecE_0,
\end{align}
which is equivalent to equation \eqref{eq:Electric_Field}. In a similar fashion one finds
\begin{align}
&X_1\vecE=s(sI-X_1 \Gamma X_1)^{-1}X_1\vecE_0,
\\\notag
&X_2\vecE=t(tI-X_2 \Gamma X_2)^{-1}X_2\vecE_0,
\\\notag
&X_1\vecJ=t(tI-X_1 \Upsilon X_1)^{-1}X_1\vecJ_0,
\\\notag
&X_2\vecJ=s(sI-X_2 \Upsilon X_2)^{-1}X_2\vecJ_0\,.	
\notag
\end{align}
These formulas determine the physical fields $\vecE$ and $\vecJ$ via
$\vecE=X_1\vecE+X_2\vecE$ and $\vecJ=\sigma_1X_1\vecE+\sigma_2X_2\vecE$.

\section{Stieltjes integrals for bulk transport coefficients}
\label{sec:Stieltjes_integrals}
We start this section with a high level description of discrete Stieltjes integral representations for the bulk transport
coefficients of uniaxial polycrystalline materials. We then utilize properties of the projection matrices in
$X_1\Gamma X_1$, for example, to provide a more detailed analysis and develop a projection method 
to show that the spectral measure can be computed by much smaller matrices, which leads to a more efficient and stable numerical algorithm for the computation of bulk transport coefficients and physical fields.

Using the the resolvent formula for $X_1\vecE$ in equation \eqref{eq:Electric_Field}, the energy functional in 
equation \eqref{eq:sigma*}, $\langle\vecE_f\rangle=0$, and $\vecE=E_0\veceh_k$ we have
\begin{align}
	\label{eq:discrete_resolvent_energy}
	\langle\vecJ\bcdot\vecE_0\rangle=
	\sigma_2E_0^2(1-\langle X_1\vecE\bcdot X_1\veceh_k \rangle/s)=
	\sigma_2E_0^2(1-\langle (sI-X_1\Gamma X_1)^{-1}X_1\veceh_k\bcdot X_1\veceh_k\rangle),
\end{align}
where $\veceh_k$ is the discretized version of the standard basis vector $\vece_k$ \cite{Murphy:2015:CMS:13:4:825}. 
Since the matrix $X_1\Gamma X_1$ is real-symmetric, it can be written \cite{Horn_Johnson-1990} as 
$X_1\Gamma X_1=W\Lambda W^T$, 
where the $\Lambda$ is a diagonal matrix with diagonal entries comprised of the real eigenvalues $\lambda_i$ 
of the matrix $X_1\Gamma X_1$ and the columns of $W$ are composed of the orthonormal eigenvectors $\vecw_i$ of the matrix, satisfying $WW^T=W^T W=I$. Therefore, we have \cite{Murphy:2015:CMS:13:4:825}
\begin{align}\label{eq:Fkk}
	F_{kk}(s)=
	\langle (sI-X_1\Gamma X_1)^{-1}X_1\veceh_k\bcdot X_1\veceh_k\rangle=
	\langle (sI-\Lambda)^{-1}W^T X_1\veceh_k\bcdot W^T X_1\veceh_k\rangle\,,
\end{align}
which can be written as the summation formula for $F_{kk}(s)$ in equation \eqref{eq:Discrete_E_F(s)}.

The following theorem is a refinement of this result and provides a rigorous mathematical formulation of
integral representations for the effective parameters for finite
lattice approximations of random uniaxial polycrystalline media. 
%
\begin{theorem} \label{thm:Discrete_Spectral_Theorem_ACM_polycrystal}
	For each $\omega\in\Omega$, let $M(\omega)=W(\omega)\Lambda(\omega)\,W(\omega)$ be the eigenvalue
	decomposition of the real-symmetric matrix
	$M(\omega)=X_1(\omega)\,\Gamma\,X_1(\omega)$. Here, the columns of the matrix $W(\omega)$
	consist of the orthonormal eigenvectors $\vecwt_i(\omega)$, $i=1,\ldots,N$,
	of $M(\omega)$ and the diagonal matrix $\Lambda(\omega)={\rm diag}(\lambda_1(\omega),\ldots,\lambda_N(\omega))$
	involves its eigenvalues $\lambda_i(\omega)$. 
	Denote $Q_i=\vecwt_i\,\vecwt_i^{\;T}$ the projection matrix onto the eigen-space
	spanned by $\vecwt_i$.
	The electric field $\vecE(\omega)$ satisfies $\vecE(\omega)=\vecE_0+\vecE_f(\omega)$, with
	$\vecE_0=\langle\vecE(\omega)\rangle$ and $\Gamma\vecE(\omega)=\vecE_f(\omega)$, and the
	effective complex permittivity tensor $\beps^*$ has components
	$\epsilon^*_{jk}$,
	$j,k=1,\ldots,d$,  which satisfy       
	\begin{align}\label{eq:Stieltjes_F_Discrete_polycrystal}
		&\epsilon^*_{jk}=\epsilon_2(\delta_{jk}-F_{jk}(s)), 
		&&F_{jk}(s)=\int_0^1\frac{\d\mu_{jk}(\lambda)}{s-\lambda}\,, 
		&&\d\mu_{jk}(\lambda)=\sum_{i=1}^N\langle \delta_{\lambda_i}(\d\lambda)\;X_1\,Q_i\hat{e}_j\bcdot\hat{e}_k\rangle\,.  
	\end{align}
\end{theorem}

Theorem \ref{thm:Discrete_Spectral_Theorem_ACM_polycrystal} 
holds for both of the settings where the matrix gradient
is full rank or rank deficient.
To numerically compute the $\mu_{jk}$ a non-standard generalization of the spectral theorem for
matrices is required, due to the projective nature of the matrices $X_1$ and $\Gamma$
\cite{Murphy:2015:CMS:13:4:825}. In particular, we develop a
\emph{projection method} that shows the spectral measure $\mu_{jk}$ in
\eqref{eq:Stieltjes_F_Discrete_polycrystal} depends only on the eigenvalues and eigenvectors of the 
upper left $N_1\times N_1$ block of the matrix $R\Gamma R^T$, where $N_1=N/d$. These
submatrices are smaller by a factor of $d$, which improves the efficiency and 
numerical computations of $\mu$ by a factor of $d^3$ \cite{Demmel:1997,Trefethen:1997:NLA}.

In this section we provide the proof for Theorem
\ref{thm:Discrete_Spectral_Theorem_ACM_polycrystal} and a projection
method for a numerically efficient projection method for computation
of spectral measures and effective parameters for uniaxial polycrystalline media.
We will use the results from Section \ref{sec:field_decomp} but for notational 
simplicity we will use $\Gamma$ instead of $\Gamma_\times$. Corollary
\ref{cor:Discrete_Spectral_Theorem_ACM} below follows immediately from the proof of 
Theorem \ref{thm:Discrete_Spectral_Theorem_ACM_polycrystal} and the results of Section
\ref{sec:field_decomp}, which provides an integral representation
for the effective resistivity $\brho^*$ involving the matrix $X_2\Gamma_\bullet X_2$ and is
analogous to the representation of $\brho^*$ for the two-component composite  
setting in \cite{Murphy:2015:CMS:13:4:825}.

From the close
analogues of this polycrystalline setting with the two-component setting discussed in 
\cite{Murphy:2024:arXiv:PolycrystalTheoryContinuum}, the proof of this theorem is analogous to Theorem 2.1 in 
\cite{Murphy:2015:CMS:13:4:825}. 
To shorten the theorem proof here, we will refer to 
\cite{Murphy:2015:CMS:13:4:825} for some of the technical details.
From Section \ref{sec:field_decomp} and the paragraph in \cite{Murphy:2015:CMS:13:4:825}
containing equations (2.39) and (2.40),
with $\chi_1$ replaced by $X_1$, we just need to prove that the functional
$F_{jk}(s)=\langle(sI-X_1\Gamma X_1)^{-1}X_1\hat{e}_j\bcdot\hat{e}_k\rangle$ 
has the integral representation displayed in
equation~\eqref{eq:Stieltjes_F_Discrete_polycrystal}. In the process, we will also establish a projection method for the numerically efficient, rigorous computation of
$\mu_{jk}$. This projection method is summarized by
equations~\eqref{Pi_coordinates_E}--\eqref{eq:Fs_W1} below.

In equations \eqref{eq:discrete_resolvent_energy} and \eqref{eq:Fkk} we gave a brief description 
about how Stieltjes integral representations for the $F_{kk}(s)$ arise. We now give a more thorough 
analysis and pay close attention to the projective nature of the matrices $X_1$ and $\Gamma$.
In \cite{Murphy:2024:arXiv:PolycrystalTheoryContinuum} we defined the real-symmetric 
mutually orthogonal projection matrices $X_i$, $i=1,2$, in terms of the
\emph{spatially dependent} rotation matrix $R$ and $C=\text{diag}(1,0,\ldots,0)$, 
all matrices of size $d\times d$. The paragraph in \cite{Murphy:2015:CMS:13:4:825} 
containing equations (2.28)--(2.30) describes how to bijectively map these $d\times d$
matrices to $N\times N$ matrices that are \emph{not} spatially dependent, where $N=L^d d$.
Under this mapping, $R$ becomes a banded rotation matrix satisfying
$R^T=R^{-1}$ and $C$ becomes $C=\text{diag}(I_1,0_1,\ldots,0_1)$, where
$I_1$ and $0_1$ are the identity and null matrices of size
$N_1=L^d$, and the vector $e_1$ is mapped to $\hat{e}_1=(1,1,\ldots,1,0,0,\ldots,0)$,
with $L^d$ ones in the first components and zeros in the rest of the components.

Writing $X_1\Gamma X_1=R^T(CR\,\Gamma \,R^TC)R$ we have 
\begin{align}\label{eq:Spec_Decomp_chi_Gamma_chi_Proof}
	X_1\Gamma X_1&=
	R^T
	\left[
	\begin{array}{ccc}
		\Gamma_1&0_{12}\\
		0_{21}&0_{22}   
	\end{array}
	\right]
	R
	=
	R^T
	\left[
	\begin{array}{ccc}
		W_1\Lambda_1W_1^{\;T}&0_{12}\\
		0_{21}&0_{22}
	\end{array}
	\right]
	R
	\notag\\
	&=
	R^T
	\left[
	\begin{array}{ccc}
		W_1&0_{12}\\
		0_{21}&I_2 
	\end{array}
	\right]    
	\left[
	\begin{array}{ccc}
		\Lambda_1&0_{12}\\
		0_{21}&0_{22}
	\end{array}
	\right]    
	\left[
	\begin{array}{ccc}
		W_1^{\;T}&0_{12}\\
		0_{21}&I_2
	\end{array}
	\right]    
	R,
\end{align}
where $I_2$ is the identity matrix of size $N_2\times N_2$, with
$N_2=N-N_1=L^d(d-1)$, and $0_{ab}$ denotes a matrix of zeros of size
$N_a\times N_b$, $a,b=1,2$. Moreover, $\Gamma_1$ is the $N_1\times N_1$ upper left principal
sub-matrix of $R\,\Gamma\,R^T$ and $\Gamma_1=W_1\Lambda_1W_1^{\;T}$ is its eigenvalue
decomposition. As $\Gamma_1$ 
is a real-symmetric matrix, $W_1$ is an orthogonal
matrix~\cite{Horn_Johnson-1990}. Also, since $R\,\Gamma\,R^T$ is a
similarity transformation of a projection matrix and $C$ is a
projection matrix, $\Lambda_1$ is a diagonal matrix with entries
$\lambda_i^1\in[0,1]$, $i=1,\ldots,N_1$, along its
diagonal~\cite{Horn_Johnson-1990,Demmel:1997}.     
Consequently, equation~\eqref{eq:Spec_Decomp_chi_Gamma_chi_Proof}
implies that the 
eigenvalue decomposition of the matrix $X_1\Gamma X_1$ is given by 
%
\begin{align}\label{eq:Spec_Decomp_chi_Gamma_chi}
	X_1\Gamma X_1=W\Lambda W^{\;T},
	\qquad
	W=R^T\left[
	\begin{array}{ccc}
		W_1&0_{12}\\
		0_{21}&I_2   
	\end{array}
	\right],
	\quad
	\Lambda=\left[
	\begin{array}{ccc}
		\Lambda_1&0_{12}\\
		0_{21}&0_{22}   
	\end{array}
	\right].
\end{align}
Here, $W$ is an orthogonal matrix satisfying $W^TW=WW^T=I$, $I$ is the
identity matrix on $\mathbb{R}^N$, and $\Lambda$ is a diagonal matrix with
entries $\lambda_i\in[0,1]$, $i=1,\ldots,N$, along its diagonal.

The eigenvalue decomposition of the matrix $X_1\Gamma X_1$ in
equation~\eqref{eq:Spec_Decomp_chi_Gamma_chi} demonstrates that its
resolvent $(sI-X_1\Gamma X_1)^{-1}$ is well defined for all
$s\in\mathbb{C}\backslash[0,1]$. In particular, by the orthogonality of the
matrix $W$, it has the following useful representation
$(sI-X_1\Gamma X_1)^{-1}=W(sI-\Lambda)^{-1}W^T$, where $(sI-\Lambda)^{-1}$ is a diagonal
matrix with entries $1/(s-\lambda_i)$ along its diagonal. This, in
turn, implies that the functional
$F_{jk}(s)=\langle(sI-X_1\Gamma X_1)^{-1}X_1\hat{e}_j\bcdot\hat{e}_k\rangle$ can be written as     
\begin{align}\label{eq:Matrix_Functional_proof}
	F_{jk}(s)
	=\langle(sI-\Lambda)^{-1}\;[X_1W]^T\hat{e}_j\,\bcdot\,W^T\hat{e}_k\rangle. 
\end{align}
Since $R^T=R^{-1}$ and $X_1=R^TCR$,
equation~\eqref{eq:Spec_Decomp_chi_Gamma_chi} implies that  
\begin{align}\label{eq:Projection_Eigenspace}
	X_1W=R^T\left[
	\begin{array}{ccc}
		W_1&0_{12}\\
		0_{21}&0_{22}  
	\end{array}
	\right]
	\quad
	\Longrightarrow
	\quad
	X_1\vecw_i=
	\begin{cases}
		\vecw_i, &\text{~for~} i=1,\ldots,N_1,  \\
		0,  &\text{~otherwise}.
	\end{cases}
\end{align}
This and the formula for $W$ in~\eqref{eq:Spec_Decomp_chi_Gamma_chi}
imply that 
\begin{align}\label{eq:Weights_chi}
	[X_1W]^T\hat{e}_j\,\bcdot\,W^T\hat{e}_k=[X_1W]^T\hat{e}_j\,\bcdot\,[X_1W]^T\hat{e}_k.
\end{align}

We are ready to provide the integral representation displayed
in~\eqref{eq:Stieltjes_F_Discrete_polycrystal} for the functional $F_{jk}(s)$ in 
equation~\eqref{eq:Matrix_Functional_proof}. Denote by
$Q_i=\vecw_i\,\vecw_i^{\;T}$, $i=1,\ldots,N$, the symmetric, mutually
orthogonal projection matrices, $Q_\ell\,Q_m=Q_\ell\,\delta_{\ell m}$, onto the eigen-spaces
spanned by the orthonormal eigenvectors
$\vecw_i$. Equation~\eqref{eq:Projection_Eigenspace} implies that 
$X_1Q_i=Q_iX_1=X_1Q_iX_1$, as $X_1Q_i=Q_i$ for $i=1,\ldots,N_1$, $X_1Q_i=0$
otherwise, and $X_1$ is a symmetric matrix. These properties allow us to write
the quadratic form $[X_1W]^T\hat{e}_j\,\bcdot\,[X_1W]^T\hat{e}_k$ as  
\begin{align}\label{eq:Quadratic_form}
	[X_1W]^T\hat{e}_j\,\bcdot\,[X_1W]^T\hat{e}_k =
	W^T\hat{e}_j\,\bcdot\,W^T\hat{e}_k =
	\sum_{i=1}^N
	(\vecw_i\bcdot\hat{e}_j)(\vecw_i\bcdot\hat{e}_k) =
	\sum_{i=1}^N Q_i\hat{e}_j\bcdot\hat{e}_k =
	\sum_{i=1}^N X_1Q_i\hat{e}_j\bcdot\hat{e}_k\,.
\end{align}
This and equations~\eqref{eq:Matrix_Functional_proof}
and~\eqref{eq:Weights_chi} yield 
\begin{align}\label{eq:Stieltjes_F_DiscretE}
	F_{jk}(s)=\int_0^1\frac{\d\mu_{jk}(\lambda)}{s-\lambda}\,, \quad
	\d\mu_{jk}(\lambda)=\sum_{i=1}^N\langle \delta_{\lambda_i}(\d\lambda)\;X_1\,Q_i\hat{e}_j\bcdot\hat{e}_k\rangle\,.
\end{align}
This 
concludes our proof of Theorem~\ref{thm:Discrete_Spectral_Theorem_ACM_polycrystal}

\begin{corollary}
	
	\label{cor:Discrete_Spectral_Theorem_ACM}
	For each $\omega\in\Omega$, let $M(\omega)=W(\omega)\Lambda(\omega)\,W(\omega)$ be 
	the eigenvalue decomposition of the real-symmetric matrix $M(\omega)=X_2(\omega)\,\Gamma_\bullet\,X_2(\omega)$. Here, the columns of the matrix $W(\omega)$
	consist of the orthonormal eigenvectors $\vecwt_i(\omega)$, $i=1,\ldots,N$,
	of $M(\omega)$ and the diagonal matrix $\Lambda(\omega)={\rm diag}(\lambda_1(\omega),\ldots,\lambda_N(\omega))$
	involves its eigenvalues $\lambda_i(\omega)$. 
	Denote $Q_i=\vecwt_i\,\vecwt_i^{\;T}$ the projection matrix onto the eigen-space
	spanned by $\vecwt_i$.
	The current field $\vecJ(\omega)$ satisfies $\vecJ(\omega)=\vecJ_0+\vecJ_f(\omega)$, with
	$\vecJ_0=\langle\vecJ(\omega)\rangle$ and $\Gamma_\bullet\vecJ(\omega)=\vecJ_f(\omega)$, 
	and the effective complex resistivity tensor $\brho^*$ has components
	$\rho^*_{jk}$,
	$j,k=1,\ldots,d$,  which satisfy       
	\begin{align}\label{eq:Stieltjes_F_Discrete_rho}
		&\rho^*_{jk}=\sigma_1^{-1}(\delta_{jk}-E_{jk}(s)), 
		&&E_{jk}(s)=\int_0^1\frac{\d\eta_{jk}(\lambda)}{s-\lambda}\,, 
		&&\d\eta_{jk}(\lambda)=\sum_{i=1}^N\langle \delta_{\lambda_i}(\d\lambda)\;X_2\,Q_i\hat{e}_j\bcdot\hat{e}_k\rangle\,.  
	\end{align}
\end{corollary}

\subsection{Projection method}\label{projection_method}

In this section we provide a formulation for a numerically efficient 
\emph{projection method} for computation of spectral measures and 
effective parameters for uniaxial polycrystalline media. Note that the 
sum in equation~\eqref{eq:Stieltjes_F_DiscretE} runs only over the index 
set $i=1,\ldots,N_1$, as equation~\eqref{eq:Projection_Eigenspace} implies
that the masses $X_1\,Q_i\hat{e}_j\bcdot\hat{e}_k$ of the measure $\mu_{jk}$
are zero for $i=N_1+1,\ldots,N$.  Denote by $\lambda_i^1$ and $\vecw_i^{\,1}$,
$i=1,\ldots,N_1$, the eigenvalues and eigenvectors of the $N_1\times N_1$
matrix $\Gamma_1=W_1\Lambda_1W_1^{\;T} $, defined in
equation~\eqref{eq:Spec_Decomp_chi_Gamma_chi_Proof}. Now, write        
\begin{align}\label{Pi_coordinates_E} 
	R\hat{e}_j=
	\left[
	\begin{array}{ccc}
		\hat{e}_j^{\,r_1}\\
		\hat{e}_j^{\,r_2}
	\end{array}
	\right],
\end{align}
where $\hat{e}_j^{\,r_1}\in\mathbb{R}^{N_1}$ and
$\hat{e}_j^{\,r_2}\in\mathbb{R}^{N_2}$. Therefore, writing the
matrix $X_1W$ in equation~\eqref{eq:Projection_Eigenspace} in block
diagonal form, $X_1W=R^T\text{diag}(W_1,0_{22})$, we have that
\begin{align}\label{eq:Reduced_Weights}
	[X_1W]^T\hat{e}_j\,\bcdot\,[X_1W]^T\hat{e}_k=[\text{diag}(W_1^T,0_{22})R\hat{e}_j]
	\bcdot[\text{diag}(W_1^T,0_{22})R\hat{e}_k]
	=[W_1^T\hat{e}_j^{\,r_1}]
	\bcdot[W_1^T\hat{e}_k^{\,r_1}]. 
\end{align}
Denote by $Q^1_i=\vecw^{\,1}_i[\vecw^{\,1}_i]^{\;T}$, $i=1,\ldots,N_1$,
the mutually orthogonal projection matrices,
$Q^1_\ell\,Q^1_m=Q^1_\ell\,\delta_{\ell m}$, onto the eigen-spaces spanned by the
orthonormal eigenvectors
$\vecw^{\,1}_i$. Equations~\eqref{eq:Matrix_Functional_proof},~\eqref{eq:Weights_chi},   
and~\eqref{eq:Reduced_Weights} then yield    
\begin{align}\label{eq:Fs_W1}
	F_{jk}(s)=\langle(sI_1-\Lambda_1)^{-1}[W_1^T\hat{e}_j^{\,r_1}]
	\bcdot[W_1^T\hat{e}_k^{\,r_1}]\rangle
	=\left\langle\sum_{i=1}^{N_1} 
	\frac{Q^1_i\hat{e}_j^{\,r_1}\bcdot\hat{e}_k^{\,r_1}                
	}{s-\lambda_i^1}
	\right\rangle.        
\end{align}

Equation~\eqref{eq:Fs_W1} demonstrates that only the spectral
information of the matrices $W_1$ and $\Lambda_1$ contribute to the functional
representation for $F_{jk}(s)$ in~\eqref{eq:Matrix_Functional_proof}
and its integral representation
in~\eqref{eq:Stieltjes_F_Discrete_polycrystal}. From a computational standpoint,  
this means that only the eigenvalues and eigenvectors of the $N_1\times N_1$
matrix $\Gamma_1$ need to be computed in order to compute the spectral
measures underlying the integral representations of the effective
parameters for finite lattice systems. This is extremely cost
effective as the numerical cost of finding all
the eigenvalues and eigenvectors of a real-symmetric $N\times N$ matrix is
$O(N^3)$~\cite{Demmel:1997}, so $N_1=N/d$ implies the computational cost
of the projection method is reduced by a factor of $d^3$.

\section{Numerical Results}\label{sec:Numerical_Results}
In this section, we utilize Theorem \ref{thm:Discrete_Spectral_Theorem_ACM_polycrystal} and the projection method of Section \ref{projection_method} to compute discrete, Stieltjes integral representations for the effective conductivity of two- and three-dimensional uniaxial polycrystalline media. We consider two-dimensional square polycrystalline media composed of a grid of randomly oriented square crystals as well as three-dimensional cubic media composed of stacked grids of randomly oriented cubic crystals. For simplicity, we will focus on the first diagonal component of the effective conductivity tensor and the underlying spectral measure $\mu_{11}$. The values of the conductivities for each crystallite are taken to be $\sigma_1 = 51.0741 + i45.1602$ in the $x$-direction and $\sigma_2 = 3.07 + i0.0019$ in both $y$- and $z$-directions, so that $s \approx -0.034 + i0.032$.

We now discuss our numerical method for computing spectral measures and effective transport coefficients in this matrix setting. Consider a uniaxial polycrystalline medium consisting of randomly oriented crystallites indexed by $\omega \in \Omega$ and described by the matrix $X_1(\omega) = R^T(\omega) C R(\omega)$. In particular, for fixed $\omega \in \Omega$, the measure $\mu_{kk}(\omega)$ is a weighted sum of $\delta$-measures centered at the eigenvalues $\lambda_i(\omega)$ of $M(\omega)$, $i = 1,...,N$, with weights $[X_1(\omega)Q_i(\omega) \hat e_k] \cdot \hat e_k$ involving the eigenvectors $\vecw_i(\omega)$ of $M(\omega)$ via $Q_i = \vecw_i \vecw_i^T$. However, equation \ref{eq:Projection_Eigenspace} implies that the measure weights $[X_1(\omega)Q_i(\omega) \hat e_k] \cdot \hat e_k$ are identically zero for $i = 1,..., N_0(\omega)$. This was used in equation \eqref{eq:Fs_W1} to show that the measure $\mu_{kk}(\omega)$ depends only on the eigenvalues $\lambda_i^1(\omega)$, $i = 1,...,N_1(\omega)$, and the eigenvectors $\vecw_i(\omega)$ of the principal sub-matrix $\Gamma_1(\omega)$ of $M(\omega)$, and that the measure weights can be expressed more explicitly as $Q_i^1(\omega) \hat e_k^{\pi_1} \cdot \hat e_k^{\pi_1}$ with $Q_i^1 = \vecw_i^1 [\vecw_i^1]^T$. Consequently, for fixed $s \in \mathbb{C}\backslash [0,1]$, the value of the effective complex conductivity $\sigma_{kk}^* = \sigma_2(1-F_{kk}(s))$ of the medium can be obtained by computing $\lambda_i^1(\omega)$ and $\vecw_i^1(\omega)$ for all $i = 1, ..., N_1(\omega)$ for each $\omega \in \Omega$. Since the computational cost of finding all the eigenvalues and eigenvectors of a $N \times N$ real-symmetric matrix is $O(N^3)$, this ``projection method" makes the numerical computation of $\mu_{kk}$ and $\sigma^*_{kk}$ more efficient by a factor of $1/d$. 

For the randomly oriented polycrystalline media described above, the cardinality $|\Omega|$ of the sample space $\Omega$ of geometric configurations is infinite due to the continuous distribution that each crystallite orientation is sampled from. In our numerical computations, every crystallite orientation $\theta_j$ is sampled from $U([0, 2\pi])$ measured from the $x$-axis (in 2D) and from the $z$-axis (in 3D). For a given arrangement of crystallites, the orientation angles $\theta_j$ for each crystallite constitutes a single geometric configuration $\omega \in \Omega$. For each $\omega \in \Omega$, \textit{all} of the eigenvalues and eigenvectors of the matrix $\Gamma_1(\omega)$ were computed using the MATLAB function \textit{eig()} despite the streamlined projection method described previously. This is necessary in order to calculate the corresponding current density $\vecJ$ which is detailed below.

In order to visually determine the behavior of the function $\mu_{11}(\lambda) = \langle Q(\lambda) \hat e_1, \hat e_1 \rangle_1$ underlying the spectral measure $\mu_{11}$ for a given random polycrystalline geometry, we plot a histogram representation of $\mu_{11}(\lambda)$, called the \textit{spectral function}, which we will also denote by $\mu_{11}(\lambda)$. In order to compute the spectral function $\mu_{11}(\lambda)$, the spectral interval $[0,1]$ was divided into $K$ sub-intervals $I_k$, $k = 1,...,K$, of equal length $1/K$. Next, for fixed $k$, we sum the discrete spectral measure weights $m_i = X_1 Q_i \hat e_1 \cdot \hat e_1$ where $i$ are the indices that correspond to each $\lambda_i$ that is in the interval $I_k$. Averaging this sum over the total number of geometric configurations produces the value displayed in the spectral function over the interval $I_k$. 

\begin{figure}[]              
\centering
\includegraphics[width=\textwidth]{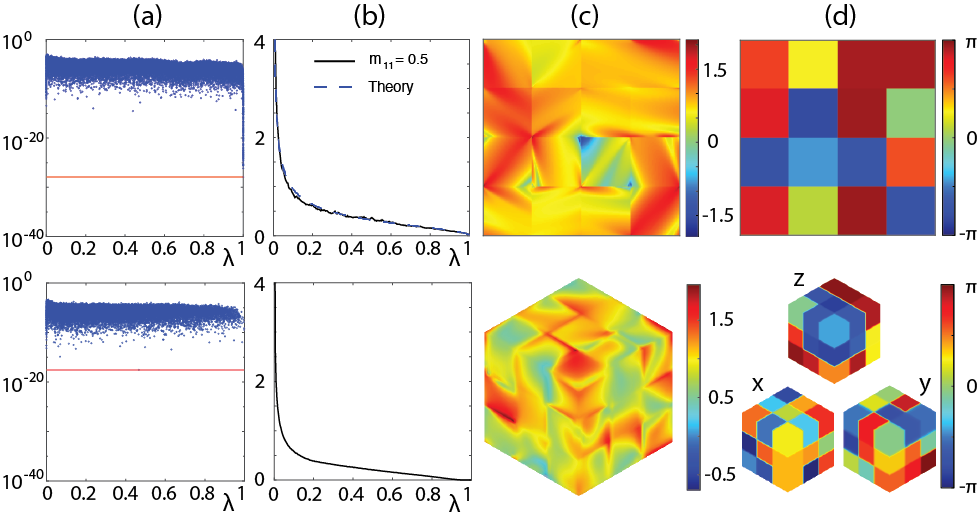}
\caption{Spectral measures and current density fields for 2D (top row) and 3D (bottom row) uniaxial polycrystalline media. (a) Spectral measure weights $m_i$ versus associated eigenvalue $\lambda_i$ of the matrix $M$ with corresponding current density field $|\vecJ|$ in log10 scale shown in (c). The red horizontal lines in (a) correspond to the smallest weight $m_i$ across the entire spectrum. The orientation of each crystallite (d) is taken to be uniformly distributed on the interval $[-\pi, \pi]$ from the $x$-axis (2D) and from the $x$, $y$, and $z$-axes (3D). The crystallite orientations in 3D are achieved through composed rotations along the $x$, $y$, and $z$ axes in a random succession.  (b) Corresponding spectral functions $\mu(\lambda)$ (histogram representations) of the spectral measure $\mu$ in (a), averaged over 2000 distinct geometric configurations. Each configuration is distinguished solely by the isotropic orientation angles $\theta_j \sim U[-\pi,\pi]$. In 2D, the geometric isotropy of the orientation angles allows the spectral function to align with the self-duality theoretical result $\mu_{kk}(\lambda)=(1/\pi)\sqrt{(1-\lambda)/\lambda}$. The electrical conductivity is taken to be $\sigma_1 = 51.074+ i45.160$ in the $x$ direction and $\sigma_2 = 3.070+i0.0019$ in the $y$ and $z$ directions. $E_0$ is taken to be oriented along the $y$-axis.}
\label{fig:figure1}
\end{figure}

In \figref{fig:figure1}(a), we display the spectral measure weights $m_i$ versus the associated eigenvalue $\lambda_i$ for a single geometric configuration of a square grid (top row) consisting of $16$ $(4 \times 4$) square crystallites of length $L_c=50$ and for a cubic grid (bottom row) consisting of $27$ ($3\times 3 \times 3$) cubic crystallites of length $L_c=10$. Histogram representations of the measure weights (the spectral function) is displayed in \figref{fig:figure1}(b), where the top row represents a $3\times 3$ grid of square crystallites with $L_c=25$ and the bottom row represents a $4 \times 4 \times 4$ grid of cubic crystallites of length $L_c=4$. The measure weights are averaged over $2000$ geometric configurations, resulting in smooth spectral functions. We reduce the length $L_c$ for these histogram representations with respect to \figref{fig:figure1}(a) due to the large number of geometric configurations and requisite computations required. 

\figref{fig:figure1}(c) displays the current density $\vecJ = \bsig \vecE = \sigma_1 X_1 \vecE + \sigma_2 X_2 \vecE$ with $\vecE_0$ taken to be $\hat e_2$ (oriented along the $y$ axis). The components $X_1 \vecE$ and $X_2 \vecE$ are computing using Equation \eqref{eq:discrete_resolvent}. The geometric configuration and orientation angles $\theta_j$ are identical to the configuration and angles taken in Figure 1(a). The orientation angles $\theta \in [-\pi,\pi]$ of each crystallite are given in \figref{fig:figure1}(d).


\section{Conclusions}
We formulated a rigorous mathematical framework that provides Stieltjes integral representations for the bulk transport coefficients of discretized uniaxial polycrystalline materials in a setting suitable for direct numerical calculation. These computations involve the spectral measures of real, symmetric random matrices of the form $X_j \Gamma X_j$ and $X_j \Upsilon X_j$, $j = 1,2$, where $\Gamma$ and $\Upsilon$ are projections onto the range of the gradient and curl operators and are constructed from finite difference matrices. Meanwhile, $X_1$ is a random, real, symmetric projection matrix which encodes crystallite arrangement and orientation. 

We show that the discrete mathematical framework describing effective transport for uniaxial polycrystalline materials closely parallels the continuum theory and showed the existence of transport fields $\vecE$ and $\vecJ$ in discrete media follows from the fundamental theorem of linear algebra and the known orthogonality properties of the domains, ranges, and kernels of finite difference representations of the curl, gradient, and divergence operators. Resolvent representations of the transport fields enable explicit computation of $\vecJ$ in arbitrary polycrystalline geometries through eigenvalue decomposition of the matrices $X_1 \Gamma X_1$ and $X_2 \Gamma X_2$. 

We developed an efficient projection method that demonstrates the spectral measures can be computed using matrices by a factor of $d$ in each dimension, leading to considerable computational savings. We validated our framework through numerical calculations of spectral measures and the effective conductivity for both 2D and 3D uniaxial polycrystalline media with checkerboard microgeometry. Our results for the spectral function in the 2D case for isotropic crystallite orientations align closely with known self-duality results.

\medskip

{\bf Acknowledgements.}
We gratefully acknowledge support from the 
Division of Mathematical Sciences at the US National Science 
Foundation (NSF) through Grants DMS-0940249, DMS-1413454,
DMS-1715680, DMS-2111117, DMS-2136198, and DMS-2206171.
We are also grateful for support from the Applied and 
Computational Analysis Program 
and the Arctic and Global Prediction Program 
at the US Office of Naval 
Research through grants 
N00014-13-1-0291,
N00014-18-1-2552,
N00014-18-1-2041
and
N00014-21-1-2909.

\medskip

\bibliographystyle{plain}
\bibliography{murphy}
\end{document}